\documentclass{article}

\usepackage[margin=1.125in]{geometry}
\usepackage{xcolor}
\definecolor{DarkGreen}{rgb}{0.1,0.5,0.1}
\definecolor{DarkRed}{rgb}{0.5,0.1,0.1}
\definecolor{DarkBlue}{rgb}{0.1,0.1,0.5}

\usepackage[small]{caption}
\usepackage[pdftex]{hyperref}
\hypersetup{
    unicode=false,          
    pdftoolbar=true,        
    pdfmenubar=true,        
    pdffitwindow=false,      
    pdfnewwindow=true,      
    colorlinks=true,       
    linkcolor=DarkBlue,          
    citecolor=DarkGreen,        
    filecolor=DarkGreen,      
    urlcolor=DarkBlue,          
    pdftitle={},
    pdfauthor={},    
}
\usepackage{amssymb,amsmath,amsthm,amsfonts,enumitem}
\usepackage{fullpage,nicefrac,comment}
\usepackage{tikz}
\usetikzlibrary{arrows,decorations.pathmorphing,decorations.shapes,snakes,patterns,positioning,calc}
\usepackage[ruled,linesnumbered]{algorithm2e}

\newcommand{\cC}{\ensuremath{\mathcal{C}}}
\newcommand{\cU}{\ensuremath{\mathcal{U}}}
\newcommand{\cB}{\ensuremath{\mathcal{B}}}

\newcommand{\cE}{\ensuremath{\mathcal{E}}}
\newcommand{\cS}{\ensuremath{\mathcal{S}}}

\newcommand{\cL}{\ensuremath{\mathcal{L}}}
\newcommand{\cT}{\ensuremath{\mathcal{T}}}

\newcommand{\F}{{\mathbb F}}

\newcommand{\PR}[1]{{\mathbb{P}}\left\{ #1\right\}}

\newcommand{\EE}{\mathbb{E}}

\newcommand{\inabs}[1]{\left|#1\right|}

\newcommand{\inset}[1]{\left\{#1\right\}}

\newcommand{\inparen}[1]{\left(#1\right)}

\newcommand{\suchthat}{\,:\,}

\newcommand{\poly}{\mathrm{poly}}

\newcommand{\bigoh}{\mathcal{O}}

\newcommand{\eps}{\varepsilon}
\renewcommand{\epsilon}{\varepsilon}

\newtheorem{theorem}{Theorem} 
\newtheorem{lemma}[theorem]{Lemma} 
\newtheorem{definition}{Definition}

\newtheorem{remark}{Remark}

\newcommand{\lc}{\text{LC}}
\newcommand{\LC}{\text{LC}}

\renewcommand{\phi}{\varphi}
\newcommand{\UI}[2]{\ensuremath{#1^{(#2)}}}
\newcommand{\ie}{\emph{i.e.,}~}
\title{Linear-time list recovery of high-rate expander codes}
\date{\today}
\author{Brett Hemenway\thanks{Computer Science Department, University of Pennsylvania.  \texttt{fbrett@cis.upenn.edu}.} \and Mary Wootters\thanks{Computer Science Department, Carnegie Mellon University.  \texttt{marykw@cs.cmu.edu}.  Research funded by NSF MSPRF grant DMS-1400558}}

\begin{document}
	\maketitle

	\begin{abstract} 	
We show that expander codes, when properly instantiated, are high-rate list recoverable codes with linear-time list recovery algorithms.  List recoverable codes have been useful recently in constructing efficiently list-decodable codes, as well as explicit constructions of matrices for compressive sensing and group testing.  Previous list recoverable codes with linear-time decoding algorithms have all had rate at most $1/2$; in contrast, our codes can have rate $1 - \eps$ for any $\eps > 0$.  We can plug our high-rate codes into a construction of Meir (2014) to obtain linear-time list recoverable codes of arbitrary rates $R$, which approach the optimal trade-off between the number of non-trivial lists provided and the rate of the code.

While list-recovery is interesting on its own, our primary motivation is applications to list-decoding.  A slight strengthening of our result would implies linear-time and optimally list-decodable codes for all rates.  Thus, our result is a step in the direction of solving this important problem.

\end{abstract}

	\section{Introduction}
		In the theory of error correcting codes, one seeks a code $\cC \subset \F^n$ so that it is possible to recover any \em codeword \em $c \in \cC$ given a corrupted version of that codeword.   
	The most standard model of corruption is from errors: some constant fraction of the symbols of a codeword might be adversarially changed.  
	Another model of corruption is that there is some uncertainty: in each position $i \in [n]$, there is some small list $S_i \subset \F$ of possible symbols.  
	In this model of corruption, we cannot hope to recover $c$ exactly; indeed, suppose that $S_i = \inset{ c_i, c'_i }$ for some codewords $c,c' \in \cC$.  
	However, we can hope to recover a short list of codewords that contains $c$.  Such a guarantee is called \em list recoverability. \em

	While this model is interesting on its own---there are several settings in which this sort of uncertainty may arise---one of our main motivations for studying list-recovery is \em list-decoding. \em  We elaborate on this more in Section~\ref{ssec:lr} below.

	We study the list recoverability of \em expander codes. \em  These codes---introduced by Sipser and Spielman in~\cite{SS96}---are formed from an expander graph and an inner code $\cC_0$.  
	One way to think about expander codes is that they preserve some property of $\cC_0$, but have some additional useful structure.  
	For example, \cite{SS96} showed that if $\cC_0$ has good distance, then so does the the expander code; the additional structure of the expander allows for a linear-time decoding algorithm.  
	In \cite{HOW13}, it was shown that if $\cC_0$ has some good (but not great) locality properties, then the larger expander code is a good locally correctable code.  
	In this work, we extend this list of useful properties to include list recoverability.  We show that if $\cC_0$ is a list recoverable code, then the resulting expander code is again list recoverable, but with a linear-time list recovery algorithm.

\subsection{List recovery}\label{ssec:lr}
	List recoverable codes were first studied in the context of list-decoding and soft-decoding: a list recovery algorithm is at the heart of the celebrated Guruswami-Sudan list-decoder for Reed-Solomon codes~\cite{GS99} and for related codes~\cite{GR08}.   Guruswami and Indyk showed how to use list recoverable codes to obtain good list- and uniquely-decodable codes~\cite{GI02,GI03,GI04a}.  More recently, list recoverable codes have been studied as interesting objects in their own right, and have found several algorithmic applications, in areas such as compressed sensing and group testing~\cite{NPR12,INR10,GNPRS13}.

	We consider list recovery from erasures, which was also studied
	in~\cite{G03a,GI04a}.  That is, some fraction of symbols
	may have no information; equivalently, $S_i = \F$ for a constant fraction of $i \in [n]$.  
	Another, stronger guarantee is list recovery from \emph{errors}.  That is, $c_i \not\in S_i$ for a constant fraction of $i \in [n]$.
	We do not consider this stronger guarantee here, and it is an interesting question to extend our results for erasures to errors.
	It should be noted that the problem of list recovery is interesting even when there are 
	neither errors nor erasures.  In that case, the problem is: given $S_i
	\subset \F$, find all the codewords $c \in \cC$ so that $c_i \in S_i$ for all
	$i$.  
	There are two parameters of interest.  First, the rate $R:= \log_q(|\cC|)/n$ of the code: ideally, we would like the rate to be close to $1$.  Second, the efficiency of the recovery algorithm: ideally, we would be able to perform list-recovery in time linear in $n$.  
	We survey the relevant results on list recoverable codes in Figure \ref{fig:litreview}.  
	While there are several known constructions of list recoverable codes with high rate, and there are several known constructions of list recoverable codes with linear-time decoders, 
	there are no known prior constructions of codes which achieve both at once.

	In this work, we obtain the best of both worlds, and give constructions of high-rate, linear-time list recoverable codes. 
	Additionally, our codes have constant (independent of $n$) list size and
	alphabet size.  As mentioned above, our codes are actually expander codes---in
	particular, they retain the many nice properties of expander codes: they are explicit
	linear codes which are efficiently (uniquely) decodable from a constant
	fraction of errors.

	We can use these codes, along with a construction of Meir~\cite{or}, to obtain linear-time list recoverable codes of any rate $R$, which obtain the optimal trade-off between the fraction $1 - \alpha$ of erasures and the rate $R$.  More precisely, for any $R \in [0,1], \ell \in \mathbb{N}$, and $\eta > 0$, there is some $L = L(\eta, \ell)$ so that we can construct rate $R$ codes which are $(R + \eta, \ell, L)$-list recoverable in linear time.  The fact that our codes from the previous paragraph have rate approaching $1$ is necessary for this construction.  To the best of our knowledge, linear-time list-decodable codes obtaining this trade-off were also not known.

It is worth noting that if our construction worked for list recovery from
\emph{errors}, rather than erasures, then the reduction above would obtain
linear-time list decodable codes, of rate $R$ and tolerating $1 - R - \eta$
errors.  (In fact, it would yield codes that are list-recoverable from errors,
which is a strictly stronger notion).  So far, all efficiently list-decodable
codes in this regime have polynomial-time decoding algorithms.  
In this sense, our work is a step in the direction of 
linear-time optimal list decoding, which is an important open problem in coding
theory.
\footnote{
In fact, adapting our construction to handle errors, even
if we allow polynomial-time decoding, is interesting.  First, it would give a new family of efficiently-decodable, optimally list-decodable codes, very different from the existing algebraic constructions.  Secondly, there are no known uniformly constructive explicit codes (that is, constructible in time $\poly(n) \cdot C_{\eta}$)
with both constant list-size and constant alphabet size---adapting our construction to handle errors, even with polynomial-time recovery, could resolve this.
}

\subsection{Expander codes}
	Our list recoverable codes are actually properly instantiated \em expander codes. \em 
	Expander codes are formed from a $d$-regular expander graph, and an \em inner code \em $\cC_0$ of length $d$, and are notable for their extremely fast decoding algorithms.  We give the details of the construction below in Section~\ref{sec:defs}.
	The idea of using a graph to create an error correcting code was first used by 
	Gallager~\cite{G63}, and the addition of an inner code was suggested by Tanner~\cite{T81}.
	Sipser and Spielman introduced the use of an expander graph in~\cite{SS96}.  
	There have been several improvements over the years
	by Barg and Zemor~\cite{Z01,BZ02,BZ05,BZ06}.

	Recently, Hemenway, Ostrovsky and Wootters~\cite{HOW13} showed that expander codes can also be \em locally corrected, \em matching the best-known constructions in the high-rate, high-query regime for locally-correctable codes.  That work showed that as long as the inner code exhibits suitable locality, then the overall expander code does as well.  This raised a question: what other properties of the inner code does an expander code preserve?  In this work, we show that as long as the inner code is list recoverable (even without an efficient algorithm), then the expander code itself is list recoverable, but with an extremely fast decoding algorithm.

	It should be noted that the works of Guruswami and Indyk cited above on linear-time list recovery are also based on expander graphs.  However, that construction is different from the expander codes of Sipser and Spielman.  In particular, it does not seem that the Guruswami-Indyk construction can achieve a high rate while maintaining list recoverability.

\subsection{Our contributions}
We summarize our contributions below:
\begin{enumerate}
	\item \textbf{The first construction of linear-time list-recoverable codes with rate approaching $1$. }  
As shown in Figure~\ref{fig:litreview}, existing constructions have either low rate or substantially super-linear recovery time.
The fact that our codes have rate approaching $1$ allows us to plug them into a construction of~\cite{or}, to achieve the next bullet point:
	\item \textbf{The first construction of linear-time list-recoverable codes with optimal rate/erasure trade-off.}
We will show in Section~\ref{sec:or} that our high-rate codes can be used to construct list-recoverable codes of arbitrary rates $R$, where we are given information about only an $R+\eps$ fraction of the symbols.
As shown in Figure~\ref{fig:litreview}, existing constructions which achieve this trade-off have substantially super-linear recovery time.
	\item \textbf{A step towards linear-time, optimally list decodable codes.}  Our results above are for list-recovery from \em erasures. \em While this has been studied before~\cite{GI04a}, it is a weaker model than a standard model which considers \emph{errors}.  As mentioned above, a solution in this more difficult model would lead to algorithmic improvements in list decoding (as well as potentially in compressed sensing, group testing, and related areas).  It is our hope that understanding the erasure model will lead to a better understanding of the error model, and that our results will lead to improved list decodable codes.
	\item \textbf{New tricks for expander codes.}
One take-away of our work is that expander codes are extremely flexible.  This gives a third example (after unique- and local- decoding) of the expander-code construction taking an inner code with some property and making that property efficiently exploitable.  We think that this take-away is an important observation, worthy of its own bullet point.  It is a very interesting question what other properties this may work for.
\end{enumerate}

	\begin{center}
		
\newcommand{\cmp}[2]{\begin{minipage}{#1in}\begin{center}#2\end{center}\end{minipage}}

\begin{figure}
\small
\begin{tabular}{|p{1.5in}|c|c|c|c|c|c|}
\hline
Source 	& Rate & List size & Alphabet & Agreement & Recovery & Explicit \\
		& & $L$ & size & $\alpha$ & time & Linear   \\
\hline\hline
Random code & $1 - \gamma$& $O(\ell/\gamma)$  & $\ell^{O(1/\gamma)}$ & $1 - O(\gamma)$ &  &  \\
\cmp{1.5}{Random pseudolinear code \cite{GI01}} & $1 - \gamma$ & $O\inparen{\frac{\ell \log(\ell)}{\gamma^2}}$ & $\ell^{O(1/\gamma)}$ & $1 - O(\gamma)$ &   &  \\
\hline
\cmp{1.5}{Random linear code \cite{venkat-thesis}} & $1 - \gamma$ & $\ell^{O(\ell/\gamma^2)}$ & $\ell^{O(1/\gamma)}$ & $1 - O(\gamma)$ &  & L \\
\hline\hline
\begin{minipage}{1.5in}\begin{center}\ \\
Folded Reed-Solomon codes~\cite{GR08} \\ \ 
\end{center}\end{minipage} 
&  $1 - \gamma$ & $n^{ O(\log(\ell)/\gamma )}$ & $ n^{O(\log(\ell)/\gamma^2)}$ & $1 - O(\gamma)$ & $n^{O( \log(\ell)/\gamma^2)}$ & EL\\
\hline
\begin{minipage}{1.5in}\begin{center}\ \\
Folded RS subcodes: evaluation points in an explicit subspace-evasive set~\cite{DL12} \\  
\end{center}\end{minipage}
& 
  $1 - \gamma$ & $(1/\gamma)^{ O(\ell/\gamma )}$ & $ n^{O(\ell/\gamma^2)}$ & $1 - O(\gamma)$ & $n^{O( \ell/\gamma^2)}$ & E\\
\hline
\begin{minipage}{1.5in}\begin{center}\ \\
Folded RS subcodes: evaluation points in a non-explicit subspace-evasive set~\cite{G11}  \\  
\end{center}\end{minipage}
&
  $1 - \gamma$ & $O\inparen{ \frac{\ell}{\gamma^2}} $ & $ n^{O(\ell/\gamma^2)}$ & $1 - O(\gamma)$ & $n^{O( \ell/\gamma^2)}$ & \\
\hline
\cmp{1.5}{\ \\ (Folded) AG subcodes~\cite{GX12,GX13}\ \\} & 1 - $\gamma$ & $O(\ell/\gamma)$ & $\exp(\tilde{O}(\ell/\gamma^2))$ & $1 - O(\gamma)$ & $C_{\ell,\gamma}n^{O(1)}$ &\\
\hline
\hline\hline
\cite{GI03} & $2^{-2^{O(\ell)}}$ & $\ell$ & $2^{2^{2^{O(\ell)}}}$ & $1 - 2^{-2^{\ell^{O(1)}}}$ & $O(n)$ & E\\
\hline
\cite{GI04a} & $\ell^{-O(1)}$ & $\ell$ & $2^{\ell^{O(1)}}$ & $.999$ $(\star)$ & $O(n)$ & E\\
\hline\hline
This work & $1 - \gamma$ & $\ell^{ \gamma^{-4} \ell^{\ell^{ C\ell/\gamma^2}}}$ & $\ell^{O(1/\gamma)}$ & $1 - O(\gamma^3)$ $(\star)$ & $O(n)$ & EL \\
\hline
\end{tabular}

\caption{Results on high-rate list recoverable codes and on linear-time
decodable list recoverable codes.  
Above, $n$ is the block length of the $(\alpha, \ell, L)$-list
recoverable code, and $\gamma > 0$ is sufficiently small and independent of $n$. Agreement rates
marked $(\star)$ are for erasures, and all others are from errors.  An empty
``recovery time" field means that there are no known efficient algorithms.
We remark that~\cite{GX13}, along with the explicit subspace designs of~\cite{GK14}, also give explicit constructions of high-rate AG subcodes
with polynomial time list-recovery and somewhat complicated parameters; the list-size $L$ becomes super-constant.
\newline\newline
The results listed above of~\cite{GR08,G11,DL12,GX12,GX13} also apply for any rate $R$ and agreement $R + \gamma$.
In Section~\ref{sec:or}, we show how to acheive the same trade-off (for erasures) in linear time using our codes.
} 
\label{fig:litreview}
\end{figure}

	\end{center}

	\section{Definitions and Notation}\label{sec:defs}
	We begin by setting notation and defining list recovery.
		An error correcting code is $(\alpha,\ell,L)$ \em list recoverable \em (from errors) if given lists of $\ell$ possible symbols 
		at every index, there are at most $L$ codewords whose symbols lie in a $\alpha$ fraction of the lists.
		We will use a slightly different definition of list recoverability, matching the definition of~\cite{GI04a}: to distinguish it from the definition above, we will call it list recoverability from \emph{erasures}.

		\begin{definition}[List recoverability from erasures]
		An error correcting code $\cC \subset \F_q^n$ is $(\alpha, \ell, L)$-list recoverable from erasures if the following holds.
		Fix any sets
		 $S_1,\ldots,S_n$ with $S_i \subset \F_q$, so that $|S_i| \le \ell$ for at least $\alpha n$ of the $i$'s and $S_i = \F_q$ for all remaining $i$.  Then there are most $L$ codewords $c \in \cC$ so that $c \in S_1 \times S_2 \times \cdots \times S_n$.
		\end{definition}
	In our study of list recoverability, it will be helpful to study the \emph{list cover} of a list $\cS \subset \F_q^n$:
		\begin{definition}[List cover]
			For a list $\cS \subset \F_q^n$, the \emph{list cover} of $\cS$ is 
			\[ 
				\lc( \cS ) = \inparen{ \inset{ c_i : c \in \cS } }_{i=1}^n. 
			\]
			The \emph{list cover size} is $\max_{i \in [n]} | \lc( \cS )_i |$.
		\end{definition}

Our construction will be based on \emph{expander graphs}.  We say a $d$-regular graph $H$ is a \emph{spectral expander} with parameter $\lambda$, if $\lambda$ is the second-largest eigenvalue of the normalized adjacency matrix of $H$.  
Intuitively, the smaller $\lambda$ is, the better connected $H$ is---see~\cite{HLW06} for a survey of expanders and their applications.
We will take $H$ to be a \em Ramanujan graph, \em  that is $\lambda \leq \frac{2\sqrt{d-1}}{d}$;
explicit constructions of Ramanujan graphs are known for arbitrarily large values of $d$ \cite{LPS88,Mar88,Mor94}.
For a graph, $H$, with vertices $V(H)$ and edges $E(H)$, we use the following notation. 
		For a set $S \subset V(H)$, we use $\Gamma(S)$ to denote the neighborhood 
		\[ \Gamma(S) = \inset{ v \suchthat \exists u \in S, (u,v) \in E(H) }. \]
		For a set of edges $F \subset E(H)$, we use $\Gamma_F(S)$ to denote the neighborhood restricted to $F$:
		\[ \Gamma_F(S) = \inset{ v \suchthat \exists u \in S, (u,v) \in F}. \]

		Given a $d$-regular $H$ and an inner code $\cC_0$, we define the \em Tanner code \em  $\cC(H, \cC_0)$ as follows.
		\begin{definition}[Tanner code~\cite{T81}]
			If $H$ is a $d$-regular graph on $n$ vertices and $\cC_0$ is a linear code of block length $d$, then 
			the \emph{Tanner code} created from $\cC_0$ and $H$ is the linear code  $\cC \subset \F_q^{E(H)}$, where each edge $H$ is 
			assigned a symbol in $\F_q$ and the edges adjacent to each vertex form a codeword in $\cC_0$.
			\[
				\cC = \{ c \in \F_q^{E(H)} \suchthat \forall v \in V(H), c|_{\Gamma(v)} \in \cC_0 \}
			\]
		\end{definition}
		Because codewords in $\cC_0$ are \emph{ordered} collections of symbols whereas edges adjacent to 
		a vertex in $H$ may be unordered, creating a Tanner code requires choosing an ordering of the edges at each vertex of the graph.
		Although different orderings lead to different codes, our results (like all previous results on Tanner codes) 
		work for all orderings.  As our constructions work with any ordering of the edges adjacent to each vertex, 
		we assume that some arbitrary ordering has been assigned, and do not discuss it further.

		When the underlying graph $H$ is an expander graph,\footnote{%
		Although many expander codes rely on bipartite expander graphs (e.g. \cite{Z01}), we find it notationally simpler to use the non-bipartite 
		version.}~%
		we call the resulting Tanner code an \emph{expander code}.
		Sipser and Spielman showed that expander codes are efficiently uniquely decodable from about a $\delta_0^2$ fraction of errors.  
		We will only need unique decoding from erasures; the same bound of $\delta_0^2$ obviously holds for erasures as well, but for completeness we state the following lemma, which we prove in Appendix~\ref{sec:erasures}. 
\begin{lemma}\label{thm:unique}
			If $\cC_0$ is a linear code of block length $d$ that can recover from an $\delta_0 d$ number of erasures, 
			and $H$ is a $d$-regular expander with normalized second eigenvalue $\lambda$, 
			then the expander code $\cC$ can be recovered from a $\frac{\delta_0}{k}$ fraction of erasures in linear time 
			whenever $\lambda < \delta_0 - \frac{2}{k}$.
		\end{lemma}
		Throughout this work, $\cC_0 \subset \F_q^d$ will be $(\alpha_0, \ell, L)$-list recoverable from erasures, and the distance of $\cC_0$ is $\delta_0$. We choose
		$H$ to be a Ramanujan graph, and $\cC = \cC(H,\cC_0)$ will be the expander code formed from $H$ and $\cC_0$.

\section{Results and constructions}
In this section, we give an overview of our constructions and state our results.  Our main result (Theorem~\ref{thm:main}) is that list recoverable inner codes imply list recoverable expander codes.  We then instantiate this construction to obtain the high-rate list recoverable codes claimed in Figure~\ref{fig:litreview}.  Next, in Theorem~\ref{thm:approachcap} we show how to combine our codes with a construction of Meir~\cite{or} to obtain linear-time list recoverable codes which approach the optimal trade-off between $\alpha$ and $R$. 
\subsection{High-rate linear-time list recoverable codes}
		\label{sec:construction}
		Our main theorem is as follows.  
		\begin{theorem}\label{thm:main} 
			Suppose that $\cC_0$ is $(\alpha_0,\ell,L)$-list recoverable from erasures, of rate $R_0$, length $d$, and distance $\delta_0$, and suppose that
			$H$ is a $d$-regular expander graph with normalized second eigenvalue $\lambda$, if 
			\[
				\lambda < \frac{ \delta_0^2 }{ 12\ell^L  }
			\]
			Then the expander code $\cC$ formed from $\cC_0$ and $H$ has rate at least $2R_0-1$ and is  $(\alpha,\ell, L')$-list recoverable from erasures, where
			\[ 
				L' \leq \exp_\ell \inparen{ { \frac{ 72\,\ell^{2L} }{ \delta_0^2 (\delta_0 - \lambda )^2 } } }
			\]
			and $\alpha$ satisfies
			\[ 1 - \alpha \geq \inparen{1 - \alpha_0}\inparen{ \frac{\delta_0(\delta_0 - \lambda)}{6}}.  \]
			Further, the running time of the list recovery algorithm is $O_{L,\ell, \delta_0,d}(n).  $
		\end{theorem}
Above, the notation $\exp_\ell(\cdot)$ means $\ell^{(\cdot)}$.
Before we prove Theorem \ref{thm:main} and give the recovery algorithm, we show how to instantiate these codes to give the parameters claimed in Figure~\ref{fig:litreview}.
%
%
			We will use a random linear code as the inner code.  The following theorem about the list recoverability of random linear codes follows from a union bound argument (see Guruswami's thesis~\cite{venkat-thesis}).
			\begin{theorem}[\cite{venkat-thesis}]\label{thm:randlin}
			For any $q \geq 2$, for all $1 \leq \ell \leq 2$, and for all $L > \ell$, and for all $\alpha_0 \in (0,1]$, a random linear code of rate $R_0$ is $(\alpha, \ell, L)$-list recoverable, with high probability, as long as
\begin{equation}\label{eq:linearcode}
			 R_0 \geq \frac{1}{\lg(q)} \inparen{ \alpha_0 \lg(q/\ell) - H(\alpha_0) - H(\ell/q) \frac{ q }{\log_q(L+1)} } - o(1). 
\end{equation}
			\end{theorem}
For any $\gamma > 0$, and any (small constant) $\zeta > 0$, choose
\[q = \exp_\ell\inparen{\frac{1}{\zeta\gamma}} \qquad \text{and} \qquad L = \exp_\ell\inparen{ \frac{\ell}{\zeta^2\gamma^2}} 
\qquad \text{and} \qquad \alpha_0 = 1 - \gamma(1 - 3\zeta). \]
Then Theorem~\ref{thm:randlin} asserts that with high probability, a random linear code of rate $R_0 = 1 - \gamma$ is $(\alpha_0, \ell, L)$-list recoverable.
Additionally, with high probability a random linear code with the parameters above will have distance $\delta_0 = \gamma(1 + O(\gamma))$.  By the union bound there exists an inner code $\cC_0$ with both the above distance and the above list recoverability.

Plugging all this into Theorem \ref{thm:main}, we get explicit codes of rate $1 -2\gamma$ which are $(\alpha, \ell, L')$-list recoverable in linear time, for
			\[ L' = \exp_\ell \inparen{ \gamma^{-4} \exp_\ell\inparen{ \exp_\ell \inparen{ C \ell / \gamma^2 } } } \]
for some constant $C = C(\zeta)$, 
and
			\[ \alpha = 1 - \inparen{ \frac{1 - 3\zeta}{6} } \gamma^3. \]
This recovers the parameters claimed in Figure~\ref{fig:litreview}.
Above, we can choose
\[ d = O\inparen{ \frac{ \ell^{2L} }{ \gamma^4}} \]
so that the Ramanujan graph would have parameter $\lambda$ obeying the conditions of Theorem \ref{thm:main}.  Thus, when $\ell,\gamma$ are constant,
so is the degree $d$, and the running time of the recovery algorithm is linear in $n$, and thus in the block length $nd$ of the expander code.
Our construction uses an inner code with distance $\delta_0 = \gamma(1+O(\gamma))$.  It is known that if the inner code in an expander graph has distance $\delta_0$, 
the expander code has distance at least $\Omega(\delta^2)$ (see for example Lemma~\ref{thm:unique}).  Thus the distance of our construction is $\delta = \Omega(\gamma^2)$.
			\begin{remark}\label{rem:inner}
			Both the alphabet size and the list size $L'$ are
constant, if $\ell$ and $\gamma$ are constant.  However, $L'$ depends rather
badly on $\ell$, even compared to the other high-rate constructions in
Figure~\ref{fig:litreview}.  This is because the bound \eqref{eq:linearcode} is
likely not tight; it would be interesting to either improve this bound or to
give an inner code with better list size $L$.  The key restrictions for such an
inner code are that (a) the rate of the code must be close to $1$; (b) the list size $L$ must be constant, and (c) the code
must be linear.  Notice that (b) and 
(c) prevent the use of either Folded Reed-Solomon codes or their restriction to
a subspace evasive set, respectively.
			\end{remark}

\subsection{List recoverable codes approaching capacity}\label{sec:or}

We can use our list recoverable codes, along with a construction of Meir~\cite{or}, to construct codes which approach the optimal trade-off between the rate $R$ and the agreement $\alpha$. To quantify this, we state the following analog of the list-decoding capacity theorem.

\begin{theorem}[List recovery capacity theorem]\label{thm:listreccap}
	For every $R > 0$, and $L \geq \ell$, there is some code $\cC$ of rate $R$ over $\F_q$ which is $(R + \eta(\ell,L), \ell, L)$-list recoverable from erasures, 
	for any 
	\[ 
		\eta(\ell, L) \geq \frac{4\ell}{L} \qquad \text{ and } \qquad q \geq \ell^{2/\eta}. 
	\]
	Further, for any constants $\eta, R > 0$, any integer $\ell$, any code of rate $R$ which is $(R - \eta, \ell, L)$-list recoverable from erasures must have $L = q^{\Omega(n)}$.
\end{theorem}

The proof is given in Appendix~\ref{sec:listreccap}.
Although Theorem \ref{thm:listreccap} ensures the existence of certain list-recoverable codes, the proof of Theorem \ref{thm:listreccap} 
is probabilistic, and does not provide a means of efficiently identifying (or decoding) these codes.
Using the approach of \cite{or} we can turn our construction of linear-time list recoverable codes 
into list recoverable codes approaching capacity.

\begin{theorem}\label{thm:approachcap}
	For any $R > 0$, $\ell > 0$, and for all sufficiently small $\eta > 0$,
	there is some $L$, depending only on $L$ and $\eta$, and some constant $d$, depending only on $\eta$,
	so that whenever $q \geq \ell^{6/\eta}$ there is a family of $(\alpha,\ell,L)$-list recoverable codes $\cC \subset \F_{q^d}^n$ with rate at least $R$, for 
	\[ 
		\alpha = R + \eta.
	\]
	Further, these codes can be list-recovered in linear time. 
\end{theorem}
We follow the approach of~\cite{or}, which adapts a construction of~\cite{AL96} to take advantage of high-rate codes
with a desirable property.  Informally, the takeaway of~\cite{or} is that, given a family of codes with any nice property
and rate approaching $1$, one can make a family of codes with the same nice property that achieves the Singleton bound.
For completeness, we describe the approach below, and give a self-contained proof.

\begin{proof}[Proof of Theorem~\ref{thm:approachcap}]
Fix $R, \ell$, and $\eta$.
Let $\alpha = R + \eta$ as above, and suppose $q \geq \ell^{2/\eta}$.
Let $R_0 = \alpha - \frac{2\eta}{3} = R + \frac{\eta}{3}$ and $R_1 = 1 - \frac{\eta}{3}$.
We construct the code $\cC$ from three ingredients: an ``outer'' code $\cC_1$ that 
is a high-rate list recoverable code with efficient decoding, a bipartite expander, 
and a short ``inner'' code that is list recoverable.  More specifically, the construction relies on:

\begin{enumerate}
	\item 
		A high-rate outer code $\cC_1$.
		Concretely, $\cC_1$, will be our expander-based list recoverable codes guaranteed by Theorem \ref{thm:main} in Section \ref{sec:construction}.
		The code $\cC_1 \subset \F_q^m$ will be of rate $R_1 = 1 - \eta/3$, and which is $(\alpha_1, \ell_1, L_1)$-list recoverable from erasures for $\alpha_1 = 1 - O(\eta^3)$ and $L_1 = L_1(\eta, \ell_1)$ depends only on $\eta,\ell_1$.  The distance of this code is $\delta_1 = \Omega(\eta^2)$.
		Note that the block-length, $m$, is specified by the choice of $R_1$ and $\ell_1$.
	\item 
		A bipartite expander graph $G = (U,V,E)$ on $2\cdot m/(R_0d) =: 2n$ vertices, with degree $d$, which has the following property:
		for at least $\alpha_1 n$ of the vertices in $U$, 
		\[ 
			|\Gamma(u) \cap A| \geq (|A|/n - \eta/3)d, 
		\]
		for any set $A \subset V$.
		Such a graph exists with degree $d$ that depends only on $\alpha_1$ and $\eta$, and hence only on $\eta$.  

	\item 
		A code $\cC_0 \subset \F_q^d$ of rate $R_0$, 
		which is $(\alpha - \eta/3, \ell, \ell_1)$-list recoverable, where
		\[ 
			R_0 = \alpha - \frac{2\eta}{3} ,\qquad  \ell_1 = \frac{12 \ell}{\eta} 
		\]
		as in the first part of Theorem~\ref{thm:listreccap}.  
		Although the codes guaranteed by Theorem~\ref{thm:listreccap} do not come with 
		decoding algorithms, we will choose $d$ to be a constant, so the code $\cC_0$ can be list recovered in constant time by a brute-force 
		recovery algorithm.
\end{enumerate}
We remark that several ingredients of this construction share notation with ingredients of the construction in the previous section (the degree $d$, code $\cC_0$, etc), although they are different.  Because this section is entirely self-contained, we chose to overload notation to avoid excessive sub/super-scripting and hope that this does not create confusion.  The only properties of the code $\cC_1$ from the previous section we use are those that are listed in Item 1.

The success of this construction relies on the fact that the code $\cC_1$ can have rate approaching 1 (specifically, rate $1 - \eta/3$).
The efficiency of the decoding algorithm comes from the efficiency of decoding $\cC_1$: since $\cC_1$ has linear time list recovery, 
the resulting code will also have linear time list recovery.

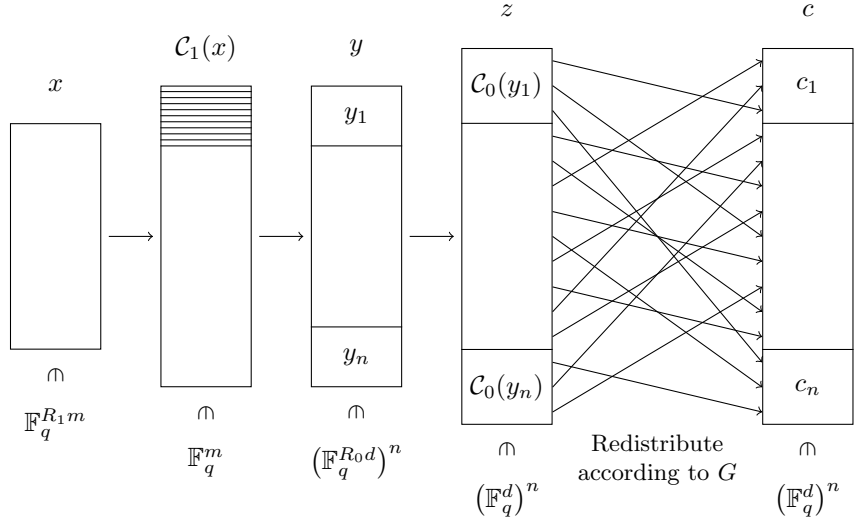
\begin{figure}
\begin{center}

\tikzstyle{codeword}=[rectangle,draw,minimum width=1.2cm]
\begin{tikzpicture}

	\coordinate (X0) at (-5,0);

	\node (X) at (-5,0) [codeword,minimum height=3cm] {};
	\node (C1X) at ([xshift=2cm]X) [codeword,minimum height=4cm] {};
	\node (Y) at ([xshift=2cm]C1X) [codeword,minimum height=4cm] {};
	\node (Z) at ([xshift=2cm]Y) [codeword,minimum height=5cm] {};
	\node (C) at ([xshift=4cm]Z) [codeword,minimum height=5cm] {};

	\node (XTOP) at ([yshift=.5cm]X.north) {$x$};
	\node (C1XTOP) at ([yshift=.5cm]C1X.north) {$\cC_1(x)$};
	\node (YTOP) at ([yshift=.5cm]Y.north) {$y$};
	\node (ZTOP) at ([yshift=.5cm]Z.north) {$z$};
	\node (CTOP) at ([yshift=.5cm]C.north) {$c$};

	\node (XBOT) at ([yshift=-1cm]X.south) {$\F_q^{R_1m}$};
	\node (C1XBOT) at ([yshift=-1cm]C1X.south) {$\F_q^m$};
	\node (YBOT) at ([yshift=-1cm]Y.south) {$\inparen{\F_q^{R_0d}}^n$};
	\node (ZBOT) at ([yshift=-1cm]Z.south) {$\inparen{\F_q^{d}}^n$};
	\node (CBOT) at ([yshift=-1cm]C.south) {$\inparen{\F_q^d}^n$};

	\path (X.south) to node [sloped] {$\in$} (XBOT);
	\path (C1X.south) to node [sloped] {$\in$} (C1XBOT);
	\path (Y.south) to node [sloped] {$\in$} (YBOT);
	\path (Z.south) to node [sloped] {$\in$} (ZBOT);
	\path (C.south) to node [sloped] {$\in$} (CBOT);

	\draw [->] ([xshift=1mm]X.east) to ([xshift=-1mm]C1X.west);	
	\draw [->] ([xshift=1mm]C1X.east) to ([xshift=-1mm]Y.west);	
	\draw [->] ([xshift=1mm]Y.east) to ([xshift=-1mm]Z.west);	

	\draw [-] ($(Y.north west)!.2!(Y.south west)$) to ($(Y.north east)!.2!(Y.south east)$);
	\node at ($(Y.north)!.1!(Y.south)$) {$y_1$};

	\draw [-] ($(Z.north west)!.2!(Z.south west)$) to ($(Z.north east)!.2!(Z.south east)$);
	\node at ($(Z.north)!.1!(Z.south)$) {$\cC_0(y_1)$};

	\draw [-] ($(C.north west)!.2!(C.south west)$) to ($(C.north east)!.2!(C.south east)$);
	\node at ($(C.north)!.1!(C.south)$) {$c_1$};

	\draw [-] ($(Y.north west)!.8!(Y.south west)$) to ($(Y.north east)!.8!(Y.south east)$);
	\node at ($(Y.north)!.9!(Y.south)$) {$y_n$};

	\draw [-] ($(Z.north west)!.8!(Z.south west)$) to ($(Z.north east)!.8!(Z.south east)$);
	\node at ($(Z.north)!.9!(Z.south)$) {$\cC_0(y_n)$};

	\draw [-] ($(C.north west)!.8!(C.south west)$) to ($(C.north east)!.8!(C.south east)$);
	\node at ($(C.north)!.9!(C.south)$) {$c_n$};

	\foreach \x in {.02,.04,.06,.08,.1,.12,.14,.16,.18,.2} {
		\draw [-] ($(C1X.north west)!\x!(C1X.south west)$) to ($(C1X.north east)!\x!(C1X.south east)$);
	}

	\foreach \a/\x in {1/.1,2/.3,3/.5,4/.7,5/.9} {
		\foreach \b/\y in {1/-.066,2/0,3/.066} {
			\pgfmathsetmacro{\z}{\x+\y}
			\pgfmathsetmacro{\w}{.2*mod( 7*(3*(\a-1)+\b-1), 5 )+.1 - \y}
			\draw [->] ($(Z.north east)!\z!(Z.south east)$) to ($(C.north west)!\w!(C.south west)$);
		}
	}

	\node [anchor=north] at ($(Z.south east)!.5!(C.south west)$) {\begin{minipage}{3.2cm} \begin{center} \small Redistribute according to $G$ \end{center} \end{minipage} };
\end{tikzpicture}
\caption{The construction of~\cite{or}. (1) Encode $x$ with $\cC_1$.  (2) Bundle symbols of $\cC_1(x)$ into groups of size $R_0d$.  (3) Encode each bundle with $\cC_0$.  (4) Redistribute according to the $n \times n$ bipartite graph $G$.  If $(u,v) \in E(G)$, and $\Gamma_i(u) = v$ and $\Gamma_j(v) = u$ then we define $c^{(v)}_j = z^{(u)}_i$.\label{fig:OM}}
\end{center}
\end{figure}

We assemble these ingredients as follows.  To encode a message $x \in \F_q^{R_1m}$, we first encode it using $\cC_1$, to obtain $y \in \F_q^m$.  
Then we break $[m]$ into $n:=m/(R_0d)$ blocks of size $R_0d$, and write $y = (y^{(1)}, \ldots, y^{(n)})$ for $y^{(i)} \in \F_q^{R_0d}$.  
We encode each part $y^{(i)}$ using $\cC_0$ to obtain $z^{(i)} \in \F_q^d$.  
Finally, we ``redistribute" the symbols of $z = (z^{(1)}, \ldots, z^{(n)})$ according to the expander graph $G$ to obtain a codeword $c \in (\F_q^d)^n$ as follows. 
We identify symbols in $z$ with left-hand vertices, $U$, in $G$ and symbols in $c$ with right-hand vertices, $V$, in $G$.
For any right-hand vertex, $v \in V$ then the $v$th symbol of $c$ is 
\[
	c_v = (a_1,\ldots,a_d) \in \F_q^d
\]
the $a_i$ are defined such that if $\Gamma_i(v) = u$ and $\Gamma_j(u) = v$, then $a_i = z^{(u)}_j$.
Intuitively, the $d$ components of $z^{(u)}$ are sent out on the $d$ edges defined by $\Gamma(u)$ and the $d$ components of 
$c^{(v)}$ are the $d$ symbols coming in on the $d$ edges defined by $\Gamma(v)$.

It is easy to verify that the rate of $\cC$ is 
\[
	R = R_0\cdot R_1 = (\alpha - 2\eta/3)(1 - \eta/3) \geq \alpha - \eta.
\] 	
Next, we give the linear-time list recovery algorithm for $\cC$ and argue that it works.  
Fix a set $A \subset V$ of $\alpha n$ coordinates so that each $v \in A$ has an associated list $S_v \subset \F_q^d$ of size at most $\ell$.
First, we distribute these lists back along the expander.  
Let $B \subset U$ be the set of vertices $u$ so that 
\[
	|\Gamma(u) \cap A| \geq \inparen{ |A|/n - \eta/3 }d = (\alpha -\eta/3)d.
\]
The structure of $G$ ensures that $|B| \geq \alpha_1 n$.  
For each of the vertices $u \in B$, the corresponding codeword $z^{(u)}$ of $\cC_0$ 
has at least an $(\alpha - \eta/3)$ fraction of lists of size $\ell$.  
Thus, for each such $u \in B$, we may recover a list $T_u$ 
of at most $\ell_1$ codewords of $\cC_0$ which are candidates for $z^{(u)}$.  
Notice that because $\cC_0$ has constant size, this whole step takes time linear in $n$.  
These lists $T_u$ induce lists $T_i$ of size $\ell_1$ for at least an $\alpha_1$ fraction of the indices $i \in [m]$.
Now we use the fact that $\cC_1$ can be list recovered in linear time from
$\alpha_1 m$ such lists; this produces a list of $L_1$ possibilities for the
original message $x$, in time linear in $m$ (and hence in $n = m/(R_0d)$) where
$L_1$ depends only on $\alpha_1$ and $\ell_1$.  Tracing backwards, $\alpha_1$
depends only on $\eta$, and $\ell_1$ depends on $\ell$ and $\eta$.  Thus,
$L_1$ is a constant depending only on $\ell$ and $\eta$, as claimed.
\end{proof}

\section{Recovery procedure and proof of Theorem~\ref{thm:main}}
In the rest of the paper, we prove Theorem \ref{thm:main}, and present our algorithm.
		The list recovery algorithm is presented in Algorithm~\ref{algo:main}, and proceeds in three steps, which we discuss in the next three sections.
\begin{enumerate}
	\item	First, we list recover locally at each vertex.  We describe this first step and set up some notation in \ref{ssec:recover}.  
	\item Next, we give an algorithm that recovers a list of $\ell$ ways to choose symbols on a constant fraction of the edges of $H$, using the local information from the first step.  This is described in Section \ref{ssec:partial}, and the algorithm for this step is given as Algorithm~\ref{algo:partial}.

	\item Finally, we repeat Algorithm~\ref{algo:partial} a constant number of times (making more choices and hence increasing the list size) to form our final list.  This third step is presented in Section~\ref{ssec:together}.
\end{enumerate}

Fix a parameter $\eps > 0$ to be determined later.\footnote{For reference, we have included a table of our notation for the proof of
Theorem~\ref{thm:main} in Figure~\ref{fig:params}.}
Set 
\begin{equation}\label{eq:alpha}
	 \alpha^* = \alpha^*(\alpha_0) := 1 - \eps \ell^L \inparen{\frac{ 1 - \alpha_0 }{ 2 - \alpha_0}}.
\end{equation}
We assume that $\alpha \geq \alpha^*$.
We will eventually choose $\eps$ so that this requirement becomes the requirement in Theorem~\ref{thm:main}.

\begin{figure}
	\begin{center}
		\begin{tabular}{|c|p{12cm}|}
			\hline
				$(\alpha_0,\ell,L)$ & List recovery parameters of inner code $\cC_0$ \\
			\hline
				$\delta_0$ & The inner code $\cC_0$ can recover from $\delta_0$ fraction of erasures \\
			\hline
				$n$ & Number of vertices in the graph $G$\\
			\hline
				$d$ & Degree of the graph (and length of inner code) \\
			\hline
				$\lambda$ & Normalized second eigenvalue of the graph \\
			\hline
				$\eps$ & A parameter which we will choose to be $\delta_0/(2k\ell^L)$.  We will find a large subgraph $H' \subset H$ so that every equivalence class in $H'$ has size at least $\eps d$. \\
			\hline
				$k$ & \rule{0pt}{1.2em} Parameter such that $k > \frac{2}{\delta_0 - \lambda}$ \\
			\hline
				$\alpha$ & The final expander code is list recoverable from an $\alpha$ fraction of erasures\\
			\hline
				$\alpha^*$ & \rule{0pt}{1.2em} Bound on the agreement $\alpha$. We set  $\alpha^* = 1 - \eps \ell^L \inparen{ \frac{1 - \alpha_0}{2- \alpha_0} }$; we assume $\alpha \geq \alpha^*$ \\
			\hline
				$\beta$ & \rule{0pt}{1.2em} Bound on the fraction of bad vertices. We set $\beta = 1 - \frac{1 - \alpha*}{2(1-\alpha_0)}$ \\
			\hline
		\end{tabular}
	\caption{Glossary of notation for the proof of Theorem~\ref{thm:main}. \label{fig:params}}
	\end{center}
\end{figure}

\subsection{Local list recovery}\label{ssec:recover}
In the first part of Algorithm \ref{algo:main} we locally list recover at each ``good" vertex.
Below, we define ``good" vertices, along with some other notation which will be useful for the rest of the proof, and 
record a few consequences of this step.

For each edge $e \in E(H)$, we are given a list $\cL_e$, with the guarantee that at least an $\alpha \geq \alpha^*$ fraction of the lists $\cL_e$ are of size at most $\ell$.  
We call an edge \emph{good} if its list size is at most $\ell$, and \emph{bad} otherwise.  Thus, there are at least a
\begin{equation}\label{eq:beta}
 \beta = \beta(\alpha_0) := 1 - \frac{1 - \alpha^*}{2(1 - \alpha_0)} 
\end{equation}
fraction of vertices which have at least $\alpha_0d$ good incident edges.  Call these vertices \em good, \em and call the rest of them \em bad. \em   For a vertex $v$, define the good neighbors $G(v) \subset \Gamma(v)$ by
\[ G(v) = \begin{cases} \Gamma(v) & v \text{ is bad } \\ \inset{ u \in \Gamma(v) \suchthat (v,u) \text{ is good } } & v \text{ is good } \end{cases} \]
Now, the first step of Algorithm \ref{algo:main} is to run the list recovery algorithm for $\cC_0$ on all of the good vertices.  Notice that because $\cC_0$ has constant size, this takes constant time.  We recover lists $\cS_v$ at each good vertex $v$.  For bad vertices $v$, we set $\cS_v = \cC_0$ for notational convenience (we will never use these lists in the algorithm).
We record the properties of these lists $\cS_v$ below, and we use the shorthand $(\alpha_0,\ell,L)$-legit to describe them.
\begin{definition}\label{def:legit}
A collection $\inset{ \cS_v }_{v \in V(H)}$ of sets $\cS_v \subset \cC_0$ is $(\alpha_0, \ell, L)$-legit if the following hold.
\begin{enumerate}
	\item For at least $\beta n$ vertices $v$ (the \em good \em vertices), $|\cS_v| \leq L$.
	\item For every good vertex $v$, at most $(1 - \alpha_0)d$ indices $i \in [d]$ have list-cover size $|\LC( \cS_v )_i| \leq \ell$.
	\item There are at most a $(1 - \alpha^*) + 2(1 - \beta)$ fraction of edges which are either bad or adjacent to a bad vertex.
\end{enumerate}
Above, $\beta$ is as in Equation~\eqref{eq:beta}, and $\alpha$ satisfies the assumption \eqref{eq:alpha}.
\end{definition}
The above discussion implies that the sets $\inset{\cS_v}$ in Algorithm \ref{algo:main} are $(\alpha_0, \ell, L)$-legit. 

		\subsection{Partial recovery from lists of inner codewords}\label{ssec:partial}
	Now suppose that we have a collection of $(\alpha_0, \ell, L)$-legit sets $\inset{ \cS_v }_{v \in V(H)}$.  We would like to recover all of the codewords in $\cC(H, \cC_0)$ consistent with these lists. 
The basic observation is that choosing one symbol on one edge is likely to fix a number of other symbols at that vertex.  To formalize this, we introduce a notion of local equivalence classes at a vertex.

			\begin{definition}[Equivalence Classes of Indices]
				Let $\inset{ \cS_v }$ be $(\alpha_0, \ell, L)$-legit and fix a good vertex $v \in V(H)$.  For each $u \in G(v)$, define
	\begin{align*}
		\UI{\phi}{u} : \cS_v 	&\rightarrow \LC( \cS_v )_u \subset \F_q \\
				c		&\mapsto c_u\qquad \mbox{(the $u$th symbol of codeword $c$)}
	\end{align*}
	Define an equivalence relation on $G(v)$ by
	\[
			u \sim_v u' \Leftrightarrow \text{ there is a permutation $\pi:\F_q \rightarrow \F_q$ so that }  \pi \circ \UI{\phi}{u}= \UI{\phi}{u'}.
	\]
	For notational convenience, for $u \not\in G(v)$, we say that $u$ is equivalent$_v$ to itself and nothing else.
	Define $\cE(v,u) \subset E(H)$ to be the (local) equivalence class of edges at $v$ containing $(v,u)$:
	\[ \cE(v,u) = \inset{ (v,u') \suchthat u' \sim_v u }. \]
	For $u \not\in G(v)$, $|\cE(v,u)| = 1$, and we call this class \em trivial. \em
\end{definition}	
It is easily verified that $\sim_v$ is indeed an equivalence relation on
$\Gamma(v)$, so the equivalence classes are well-defined.  
Notice that
$\sim_v$ is specific to the vertex $v$: in particular, $\cE(v,u)$ is not
necessarily the same as $\cE(u,v)$.  
For convenience, for bad vertices $v$, we say that $\cE(v,u) = \inset{(v,u)}$ for all $u \in \Gamma(v)$ (all of the local equivalence classes at $v$ are trivial).

We observe a few facts that follow
immediately from the definition of $\sim_v$.
\begin{enumerate}
	\item  For each $u \in G(v)$, we have $| \LC(\cS_v)_u | \leq \ell$ by the
assumption that $\cS_v$ is legit.  Thus, there are most $\ell^L$ choices for
$\phi^{(u)}$, and so there are at most $\ell^L$ nontrivial equivalence classes
$\cE(v,u)$.  (That is, classes of size larger than $1$).
	\item
	The average size of a nontrivial equivalence class $\cE(v,u)$ is at least $\frac{\alpha_0 d}{\ell^L}$.	
	\item
	If $u \sim_v u'$, then for any $c \in
\cS_v$ the symbol $c_u$ determines the symbol $c_{u'}$.  Indeed, $c_u =
\pi(c_{u'})$ where $\pi$ is the permutation in the definition of $\sim_u$.
In particular, \em learning the symbol on $(v,u)$ determines the symbol on
$(v,u')$ for all $u' \sim_v u$. \em
\end{enumerate}

The idea of the partial recovery algorithm follows from this last observation.  If we pick an edge at random and assign it a value, then we expect that this determines the value of about $\alpha_0 d/\ell^L$ other edges.  These choices should propagate through the expander graph, and end up assigning a constant fraction of the edges.  We make this precise in Algorithm \ref{algo:partial}.
To make the intuition formal, we also define a notion of global equivalence between two edges.
			\begin{definition}[Global Equivalence Classes]
				For an expander code $\cC(H,\cC_0)$, and
$(\alpha_0, \ell, L)$-legit lists $\inset{ \cS_v }_{v \in V(H)}$, we define an equivalence relation $\sim$ as follows.  For good vertices $a,b,u,v$, we say
				\[
					(a,b) \sim (u,v)
				\]
				if there exists a path from $(a,b)$ to $(u,v)$ where each adjacent pair of edges is in its local equivalence relation, \ie
	there exists $(w_0=a,w_1=b),(w_1,w_2),\ldots,(w_{n-2},w_{n-1}), (w_{n-1}=u,w_n=v)$  so that
				\[
					(w_i,w_{i+1}) \in \cE(w_{i+1},w_{i+2}) \mbox{ for $i = 0,\ldots,n-2$ } \qquad \text{and} \qquad (w_i, w_{i+1}) \text{ is good for all $i = 0,\ldots, n-1$.}
				\]
				Let $\cE^{H}_{(u,v)} \subset E(H)$ denote the global equivalence class of the edge $(u,v)$.
			\end{definition}
It is not hard to check that this indeed forms an equivalence relation on the
edges of $H$ and that a single decision about which of the $\ell$ symbols
appears on an edge $(u,v)$ forces the assignment of all edges in
$\cE^{H}_{(u,v)}$.   
			\begin{algorithm}
			\KwIn{Lists $\cS_v \subset \cC_0$ which are $(\alpha_0, \ell, L)$-legit, and a starting good edge $(u,v)$, where both $u$ and $v$ are good vertices. }
			\KwOut{A collection of at most $\ell$ partial assignments $\UI{x}{\sigma} \in \inparen{ \F_q \cup \inset{\bot} }^{E(H)}$,
			for $\sigma \in LC(\cS_v)_u \cap \LC(\cS_u)_v $. }
			
			\For{ $\sigma \in \LC(\cS_v)_u \cap \LC( \cS_u)_v$ }{
				Initialize $\UI{x}{\sigma} = (\bot,\ldots,\bot) \in \inparen{ \F_q \cup \inset{\bot} }^{E(H)}$
				
				\For{ $(v,u') \in \cE(v,u)$ }{
					Set $\UI{x}{\sigma}_{(v,u')}$ to the only value consistent with the assignment $\UI{x}{\sigma}_{(v,u)} = \sigma$
				}
				\For{ $(v',u) \in \cE(u,v)$ }{
					Set $\UI{x}{\sigma}_{(v',u)}$ to the only value consistent with the assignment $\UI{x}{\sigma}_{(v,u)} = \sigma$
				}
				
				
				Initialize a list $W_0 = \inset{u,v}$. 
				
				\For{ $t=0,1,\ldots$ }
				{
					\textbf{If} $W_t = \emptyset$, break.
					
					$W_{t+1} = \emptyset$
					
					\For{ $a \in W_t$ }
					{
						\For { $b \in \Gamma(a)$ where $\UI{x}{\sigma}_{(a,b)} \neq \bot$ }
						{
							Add $b$ to $W_{t+1}$\\
							\For{ $(b,c) \in \cE(b,a)$ }
							{
								Set $\UI{x}{\sigma}_{(b,c)}$ to the only value consistent with the assignment $c_{(a,b)} = \UI{x}{\sigma}_{(a,b)}$
								\label{line:assign}
							}
						}
					}
				}
			}	
			\Return{$\UI{x}{\sigma}$ for $\sigma \in \cS_v(u)$.}
				
			\caption{Partial decision algorithm}\label{algo:partial}
			\end{algorithm}
			
			Algorithm \ref{algo:partial} takes an edge $(u,v)$ and iterates through all $\ell$ possible assignments 
			to $(u,v)$ and turns this into $\ell$ possible assignments for the vectors $\langle c_e \rangle_{e \in \cE^{H}_{(u,v)}}$.
			In order for Algorithm \ref{algo:partial} to be useful, the graph $H$ should have some large equivalence classes.
			Since each good vertex has at most $\ell^L$ nontrivial equivalence classes which partition its $\geq \alpha_0 d$ good edges,
			most of the nontrivial local equivalence classes are larger than $\frac{\alpha_0 d}{\ell^L}$.
			This means that a large fraction of the edges are themselves contained in large local equivalence classes.
			This is formalized in Lemma \ref{lem:subgraph}.

			\begin{lemma}\label{lem:subgraph} Suppose that $\inset{ \cS_v }$ is $(\alpha_0, \ell, L)$-legit, and consider the local equivalence classes defined with respect to $\cS_v$.  There a large subgraph $H'$ of $H$ so that
				$H'$ contains only edges in large local equivalence classes.
				In particular,
				\begin{itemize}
					\item $V(H') = V(H),$
					\item for all $(v,u) \in E(H')$, $|\cE(v,u) \cap E(H')| \geq \eps d$,
					\item $|E(H')| \geq \inparen{ \frac{nd}{2}} \inparen{ 1 - 3 \eps \ell^L }$.
				\end{itemize}
			\end{lemma}

			\begin{proof}
				Consider the following process:  
				\begin{itemize}
				\item Remove all of the bad edges from $H$, and remove all of the edges incident to a bad vertex.
				\item While there are any vertices $v,u \in V(H)$ with $|\cE(v,u)| < \eps d$:
					\begin{itemize}
						\item  Delete all classes $\cE(v,u)$ with $|\cE(v,u)| < \eps d$. 
					\end{itemize}
				\end{itemize}
				We claim that the above process removes at most a 
				\[ 2\ell^L \eps + (1 - \alpha^*) + 2(1 - \beta)\]
				fraction of
				edges from $H$.  By the definition of $(\alpha_0, \ell,L)$-legit, there are at most $(1 - \alpha^*) + 2(1 - \beta)$ fraction removed in the first step.  To analyze the second step, 
				call a good vertex $v \in V(H)$ \em active \em in a round if we remove $\cE(v,u)$.
				Each good vertex is active at most $\ell^L$ times, because there are at most $\ell^L$ nontrivial classes $\cE(v,u)$ for every $v$ (and we have already removed all of the trivial classes in the first step).
				At each good vertex, every time it is active, we delete at most $\eps d$ edges.  Thus, we have deleted a total of at most
				\[  n \cdot \ell^L \cdot \eps d \]
				edges, and this proves the claim.  
			Finally, we observe that our choice of $\alpha^*$ and $\beta$ in \eqref{eq:alpha} and \eqref{eq:beta} respectively implies that $(1 - \alpha^*) + 2(1 - \beta) \leq \eps \ell^L$.  
Since the remaining edges belong to classes of size at least $\eps d$, this proves the lemma.
			\end{proof}

			A basic fact about expanders is that if a subset $S$ of vertices has a significant fraction of its edges contained in $S$, 
			then $S$ itself must be large.  This is formalized in Lemma \ref{lem:expand}.

			\begin{lemma}
				\label{lem:expand}
				Let $H$ be a $d$-regular expander graph with normalized second eigenvalue $\lambda$.
				Let $S \subset V(H)$ with $|S| < (\eps - \lambda)n$, and $F \subset E(H)$ so that for all $v \in S$, 
				\[ |\inset{ e \in F : e \text{ is adjacent to } v }| \geq \eps d. \]
				Then
				\[ |\Gamma_F(S)| > |S|. \]
			\end{lemma} 

			\begin{proof} 
				The proof follows from the expander mixing lemma.  Let $T = \Gamma_F(S)$.  Then 
				\begin{align*} 
					\eps d |S| &\leq E(S,T)\\
					&\leq \frac{d|S||T|}{n} + d\lambda \sqrt{ |S||T| } \\
					&\leq \frac{d|S||T|}{n} + \frac{ d\lambda \inparen{ |S| + |T| } }{2}.
				\end{align*}
				Thus, we have
				\[ |T| \geq |S| \inparen{ \frac{ \eps - \lambda/2 } { |S|/n + \lambda/2 } }. \]
				In particular, as long as $|S| < n(\eps -\lambda)$, we have
				\[ |T| > |S|. \]
			\end{proof}

			\begin{lemma}[Expanders have large global equivalence classes]
				\label{lem:globalequivalence}
				If $(u,v)$ is sampled uniformly from $E(H)$, then
				\[
					\Pr_{(u,v)} \left[ | \cE^{H}_{(u,v)} | > \frac{nd}{2} \eps (\eps -\lambda) \right] > 1 - 3 \eps \ell^L
				\]
			\end{lemma}

			\begin{proof}
From Lemma \ref{lem:subgraph}, there is a subgraph $H' \subset H$ such that $|E(H')| > |E(H)|(1 - 3\eps \ell^L)$.
Let $(u,v) \in E(H')$, and consider $\cE^{H'}_{(u,v)}$.  Let $S$ be the set of vertices in $H'$ that are adjacent to an edge in $\cE^{H'}_{(u,v)}$, 
				\ie
				\[
					S = \{ w \in V(H') | (w,z) \in \cE^{H'}_{(u,v)} \mbox{ for some } z \in V(H') \}
				\]
				Since every local equivalence class in $H'$ is of sized at least $\eps d$, then every vertex in $S$ has at least 
				$\eps d$ edges in $\cE^{H'}_{(u,v)}$.  By definition of $S$, for every edge in $\cE^{H'}_{(u,v)}$ both its endpoints are in $S$
				and $\Gamma_{\cE^{H'}_{(u,v)}} (S) = S$.  Thus by Lemma \ref{lem:expand}, it must be that $|S| \ge n(\eps - \lambda)$.
				Thus 
				\[
					|\cE^{H'}_{(u,v)}| \ge \frac{\eps d}{2} |S| = \frac{nd}{2} \eps (\eps - \lambda)
				\]			
				Then any edge in $H'$ is contained in an equivalence class of size at least $\frac{nd}{2} \eps(\eps - \lambda)$, 
				and the result follows form the fact that $|H'| > (1- 3 \eps \ell^L) |H|$.
			\end{proof}

Finally, we are in a position to prove that Algorithm \ref{algo:partial} does what it's supposed to.
			\begin{lemma} 
				\label{lem:partial}
				Algorithm \ref{algo:partial} produces a list of at most $\ell$ partial assignments $x^{(\sigma)}$ 
				so that 
				\begin{enumerate}
					\item
						Each of these partial assignments assigns values to the same set.  Further, this set is the global equivalence class $\cE^{H}_{(v,u)}$, where $(v,u)$ is the initial edge given as input 
						to Algorithm \ref{algo:partial}.
					\item
						For at least $(1-3 \eps \ell^L)$ fraction of initial edges $(v,u)$ we have 
						$|\cE^{H}_{(v,u)}| \ge \eps(\eps - \lambda) |E(H)|$
					\item
						Algorithm \ref{algo:partial} can be implemented so that the running time is $O_{\ell,L}(|\cE^{H}_{(v,u)}|)$.
				\end{enumerate}
			\end{lemma}

			\begin{proof}
				For the first point, notice that at each $t$ in
Algorithm \ref{algo:partial}, the algorithm looks at each vertex it visited in
the last round (the set $W_{t-1}$) and assigns values to all edges in the
(local) equivalence classes $\cE(v,u)$ for all $v \in W_{t-1}$ for which at
least one other value in $\cE(v,u)$ was known.  Thus if there is a path of
length $p$ from $(v,u)$ to $(z,w)$ walking along (local) equivalence classes,
the edge $(z,w)$ will get assigned by Algorithm \ref{algo:partial} in at most
$p$ steps.

For the second point, by Lemma \ref{lem:globalequivalence} for at least a
$(1-2\eps \ell^L)$ fraction of the initial starting edges, $(v,u)$ we have
$|\cE^{H}_{(v,u)}| \ge \eps(\eps - \lambda) |E(H)|$.

Finally, we remark on running time.  Since each edge is only in two (local)
equivalence classes (one for each vertex), it can only be assigned twice during
the running of the algorithm (Algorithm \ref{algo:partial} line
\ref{line:assign}).  Since each edge can only be assigned twice, the total
running time of the algorithm will be $O(|\cE^{H}_{(u,v)}|)$.
\end{proof}

\subsection{Turning partial assignments into full assignments}\label{ssec:together}
Given lists $\{\cL_e\}_{e \in E(H)}$, Algorithm~\ref{algo:main} runs the local recovery algorithm at each good vertex $v$ to obtain $(\alpha_0, \ell,L)$-legit sets $\cS_v$.  Then 
Given $(\alpha_0,\ell,L)$-legit sets $\cS_v$, Algorithm~\ref{algo:partial} can find $\ell$ partial codewords, defined on $\cE^H_{(u,v)}$.  
			To turn this into a full list recovery algorithm, we simply need to run Algorithm \ref{algo:partial} multiple times, obtaining partial 
			assignments on disjoint equivalence classes, and then stitch these partial assignments together.  If we run Algorithm \ref{algo:partial} 
			$t$ times, then stitching the lists together we will obtain at most $\ell^t$ possible codewords; this will give us our final list of size $L'$.
			This process is formalized in Algorithm \ref{algo:main}.

			\begin{algorithm}
				\KwIn{ A collection of lists $\cL_e \subset \F_q$ for at least a $\alpha^*$ fraction of the $e \in E(H)$, $|\cL_e| \leq \ell$.}
				\KwOut{  A list of $\cL'$ assignments $c \in \cC$ that are consistent with all of the lists $\cL_e$. }

				Divide the vertices and edges of $H$ into good and bad vertices and edges, as per Section~\ref{ssec:recover}.
				Run the list recovery algorithm of $\cC_0$ at each good vertex $v \in V(H)$ on the lists
				$\inset{ \cL_{(v,u)} : u \in \Gamma(v) }$
				to obtain $(\alpha_0, \ell, L)$-legit lists $\cS_v \subset \cC_0$.  

				Initialize the set of \texttt{unassigned} edges $\cU = E(H)$.

				Initialize the set of \texttt{bad} edges $\cB$ to be the bad edges along with the edges adjacent to bad vertices.

				Initialize $\mathcal{T} = \emptyset$.

				\For {$t=1,2,\ldots$ }
				{
					\If{ $\cU \subset \cB$ }{break}

					Choose an edge $(v,u) \gets \cU \setminus \cB$.
					
					Run Algorithm \ref{algo:partial} on the collection $\inset{\cS_a \suchthat a \in V(H)}$ and on the starting edge $(v,u)$.  
					This returns a list
					$x^{(t)}_1,\ldots, x^{(t)}_\ell$
					of assignments to the edges in $\cE^{H}_{(u,v)}$. \tcc*{Notice that the notation $x^{(t)}_j$ differs from that in Algorithm \ref{algo:partial}}
					
					\If {$|\cE^{H}_{(u,v)}| > \eps( \eps - \lambda) \frac{nd}{2}$}
						{Set $\cU = \cU \setminus \cE^{H}_{(u,v)}$ and set $\cT = \cT \cup \inset{t}$ }
					\Else{ $\cB = \cB \cup \cE^{H}_{(u,v)}$ }
				}

				\For{ $t \in \mathcal{T}$  and $j = 1,\ldots, \ell$}
				{
						Concatenate the (disjoint) assignments $\inset{ x^{(t)}_{j} \suchthat t \in \cT }$ to obtain an assignment $x$.
						
						Run the unique decoding algorithm for erasures (as in Lemma~\ref{thm:unique}) for the expander code to correct the partial assignment $x$ to a codeword $c \in \cC$.
						
						If $c$ agrees with the original lists $\cL_e$, add it to the output list $\cL'$.
				}
				Return $\cL' \subset \cC$.
				\caption{List recovery for expander codes} \label{algo:main}
			\end{algorithm}
The following theorem asserts that Algorithm~\ref{algo:main} works, as long as we can choose $\lambda$ and $\eps$ appropriately.
			\begin{theorem}\label{thm:itworks}
				Suppose that the inner code $\cC_0$ is $(\alpha_0, \ell, L)$-list recoverable with and 
				has distance $\delta_0$.  Let $\alpha \geq \alpha^*$ as in Equation~\eqref{eq:alpha}.
				Choose $k > 0$ so that $\lambda < \delta_0 - \frac{2}{k}$, 
				and set $\epsilon = \frac{\delta_0}{2k \ell^L}$.  Then
				Algorithm \ref{algo:main} returns a list of at most
				\[ 
					L' = \ell^{ \frac{1}{\eps(\eps - \lambda) } }
				\]
				codewords of $\cC$.
				Further, this list contains every codeword consistent with the lists $\cL_e$.  In particular, $\cC$ is $(\alpha, \ell, L')$-list recoverable from erasures.

				The running time of Algorithm \ref{algo:main} is $O_{\ell,L,\eps}(nd).$
			\end{theorem}

			\begin{proof}[Proof of Theorem \ref{thm:itworks}]
First, we verify the list size. 
				For each $t \in \cT$, Algorithm
\ref{algo:partial} covered at least a $\eps(\eps - \lambda)$ fraction of the
edges, so $|\cT| \leq \frac{1}{\eps(\eps - \lambda)}$.  Thus the number of
possible partial assignments $x$ is at most \[ \ell^{|\cT|} \leq \ell^{
\frac{1}{\eps(\eps - \lambda)}}. \] 

				Next, we verify correctness.  By Lemma
\ref{lem:globalequivalence}, at least a $1 - 3 \eps \ell^L$ fraction of the
edges are in equivalence classes of size at least $(nd/2) \eps (\eps -
\lambda)$.  Thus at the end of Algorithm \ref{algo:main} at most $3\eps \ell^L$
vertices are uncorrected.  By Lemma~\ref{thm:unique}, we can correct these
erasures in linear time this as long as $3\eps \ell^L < \delta_0/k$ (which was
our choice of $\eps$), and as long as $\lambda < \delta_0 - \frac{2}{k}$ (which
was our assumption).   Thus, Algorithm~\ref{algo:main} can uniquely complete
all of its partial assignments.  Since any codeword $c \in \cC$ which agrees
with all of the lists agrees with at least one of the partial assignments $x$,
we have found them all.

				Finally, we consider runtime.  As a
pre-processing step, Algorithm \ref{algo:main} takes $O( T_d n )$ steps to run
the inner list recovery algorithm at each vertex, where $T_d \leq O(d \ell
|\cC_0|)$ is the time it takes to list recover the inner code
$\cC_0$.\footnote{ Because $d$ is constant, we can write $T_d = O(1)$, but it
may be that $d$ is large and that there are algorithms for the inner code that
are better than brute force. } It takes another $O(dn)$ steps of preprocessing
to set up the appropriate graph data structures.  Now we come to the first
loop, over $t$.  By Lemma \ref{lem:partial}, the equivalence classes
$\cE^H_{(u,v)}$ form a partition of the edges, and at least a $(1 - 3 \eps
\ell^L)$ fraction of the edges are in parts of size at least $\frac{nd}{2} \eps (\eps -
\lambda)$.  By construction, we encounter each class only once; because the
running time of Algorithm \ref{algo:partial} is linear in the size of the part,
the total running time of this loop is $O(dn)$.  Finally, we loop through and
output the final list, which takes time $O(L' dn)$, using the fact
(Lemma~\ref{thm:unique}) that the unique decoder for expander codes runs in
linear time.
			\end{proof}

Finally, we pick parameters and show how Theorem~\ref{thm:itworks} implies Theorem~\ref{thm:main}.
\begin{proof}[Proof of Theorem \ref{thm:main} ]
	Theorem \ref{thm:main} requires choosing appropriate parameters to instantiate Algorithm \ref{algo:main}.
				In order to apply Theorem~\ref{thm:itworks}, we choose
				\[
					k > \frac{2}{\delta_0 - \lambda} > 0 \qquad \text{and} \qquad
					\eps = \frac{\delta_0}{3k \ell^L} < \frac{\delta_0}{ 3\inparen{ \frac{2}{\delta_0 - \lambda} } \ell^L } = \frac{\delta_0 (\delta_0 - \lambda) }{6 \ell^L }.
				\]
				This ensures that the hypotheses of Theorem~\ref{thm:itworks} are satisfied.  The assumption that $\lambda < \delta_0^2/(12\ell^L)$ and the bound on $\eps$ implies that
$\eps - \lambda > \eps/2.$
Thus, the conclusion of Theorem \ref{thm:itworks} about the list size reads
\[ L' \leq \exp_\ell\inparen{\frac{1}{\eps(\eps - \lambda)}} \leq \exp_\ell\inparen{ \frac{2}{\eps^2} } \leq \exp_\ell\inparen{ \frac{ 72 \ell^{2L} }{ \delta_0^2 (\delta_0 - \lambda)^2 }}. \]
The definition of $\alpha^*$ from \eqref{eq:alpha} becomes
\[
\alpha^* = 1 - \eps \ell^L\inparen{ \frac{1 - \alpha_0}{2 - \alpha_0} } \leq 1 - \frac{ \delta_0(\delta_0 - \lambda)}{6} \inparen{ \frac{1 - \alpha_0}{2 -\alpha_0}},
\]
which implies the claim about $\alpha$.
Along with the statement about running time from Theorem~\ref{thm:itworks}, this completes the proof of Theorem~\ref{thm:main}.
			\end{proof}

\section{Conclusion and open questions} 

We have shown that expander codes, properly instantiated, are high-rate list recoverable codes with constant list size and constant alphabet size, which can be list recovered in linear time.  
To the best of our knowledge, no such construction was known.

Our work leaves several open questions.  Most notably, our algorithm can handle \em
erasures, \em but it seems much more difficult to handle errors.  As mentioned above, handling
list recovery from errors would 
 open the door for many of the applications
of list recoverable codes, to list-decoding and other areas. 
Extending our results to errors with linear-time recovery would be most interesting, as it
would immediately lead to optimal linear-time list-decodable codes.  However, even
polynomial-time recovery would be interesting: in addition to given a new, very different family of efficient locally-decodable codes, this could lead to explicit (uniformly constructive), efficiently list-decodable codes with constant list size
and constant alphabet size, which is (to the best of our knowledge) currently an open problem.

Second, the parameters of our construction could be improved: our choice of inner code (a random
linear code), and its analysis, is clearly suboptimal.  Our construction would
have better performance with a better inner code.  As mentioned in
Remark~\ref{rem:inner}, we would need a high-rate linear code which is list
recoverable with constant list-size (the reason that this is not begging the
question is that this inner code need not have a fast recovery algorithm).  We
are not aware of any such constructions.

\section*{Acknowledgments} We thank Venkat Guruswami for raising the question of obtaining high-rate linear-time list-recoverable codes, and for very helpful conversations.  We also thank Or Meir for pointing out~\cite{or}.

\bibliographystyle{alpha}
\bibliography{lr.bib}

\appendix

	\section{Linear-time unique decoding from erasures}\label{sec:erasures}

		In this appendix, we include (for completeness) the algorithm for uniquely decoding an expander code from erasures, and a proof that it works.
		Suppose $\cC$ is an expander code created from a $d$-regular graph $H$ and inner code $\cC_0$ of length $d$, so that the inner code $\cC_0$ can be corrected from an $\delta_0 d$ erasures.
		
		\begin{algorithm}
			\KwIn{Input: A vector $w \in \inset{ \F_q \cup \bot }^{|E(H)|}$}
			\KwOut{Output: A codeword $c \in \cC$}
			
			Initialize $\cB_0 = V(H)$	

			\For{ $t = 1,2,\ldots$ \label{line:outerloop} }{
				\If{ $\cB_{t-1} = \emptyset$ }{ \texttt{Break} }
				$\cB_{t} = \emptyset$ \\
				\For{ $v \in \cB_{t-1}$ \label{line:innerloop}}
					{
						\If{ $| \inset{ u \suchthat  (u,v) \in E(H),  w_{(u,v)} = \bot } | < \delta_0 d $ }{ 
							Correct all $(u,v) \in \Gamma(v)$ using the local erasure recovery algorithm 
						}
						\Else{ $\cB_{t} = \cB_{t} \cup \{v\}$ }
					}
			}
			Return the corrected word $c$.
			\caption{A linear time algorithm for erasure recovery }\label{algo:erasurerecovery}
		\end{algorithm}		

		\begin{lemma}[Restatement of Lemma~\ref{thm:unique}]
			If $\cC_0$ is a linear code of block length $d$ that can recover from an $\delta_0 d$ number of erasures, 
			and $H$ is a $d$-regular expander with normalized second eigenvalue $\lambda$, 
			then the expander code $\cC$ can be recovered from a $\frac{\delta_0}{k}$ fraction of erasures in linear time using
			Algorithm \ref{algo:erasurerecovery} whenever $\lambda < \delta_0 - \frac{2}{k}$.
		\end{lemma}

		\begin{proof}
			Since there are at most $\frac{\delta_0}{k} |E(H)|$ erasures, at most $\frac{2}{k} |V(H)|$ of the nodes are adjacent 
			to at least $\delta_0 d$ erasures.  Thus $|\cB_1| \le \frac{2}{k} |V(H)|$.
			
			By the expander mixing lemma
			\begin{equation}
				\label{eqn:unknownUB}
				|E(\cB_{t-1},\cB_t)| \le \frac{d |\cB_{t-1}||\cB_t| }{n} + \lambda d \sqrt{|\cB_{t-1}|| \cB_t|}
			\end{equation}
			On the other hand, in iteration $t-1$ of the outer loop, every vertex in $\cB_{t}$ has at least $\delta_0 d$ unknown edges, 
			and these edges must connect to vertices in $\cB_{t-1}$ (since at step $t-1$ all vertices in $V(H) \setminus \cB_{t-1}$ are completely known).

			Thus
			\begin{equation}
				\label{eqn:unknownLB}
				|E(\cB_{t-1},\cB_t)| \ge \delta_0 d |\cB_t|
			\end{equation}

			Combining equations \ref{eqn:unknownUB} and \ref{eqn:unknownLB} we see that 
			\begin{align*}
				\delta_0 d |\cB_{t}| 	&\le \frac{d |\cB_{t-1}||\cB_t| }{|V(H)|} +  \lambda d \sqrt{ |\cB_{t-1}| |\cB_t| }\\
									&\Downarrow\\
						\delta_0  		&\le \frac{|\cB_{t-1}|}{|V(H)|} + \lambda \sqrt{ \frac{|\cB_{t-1}|}{|\cB_t|} } \\
									&\Downarrow\\
						|\cB_{t}|	&\le \frac{ \lambda^2 |\cB_{t-1}|}{\inparen{ \delta_0 - \frac{|\cB_{t-1}|}{|V(H)|} }^2 } \\
									&\Downarrow\\
						|\cB_{t}|	&\le \frac{ \lambda^2 |\cB_{t-1}|}{\inparen{ \delta_0 - \frac{2}{k} }^2 } \\
			\end{align*}

			Where the last line uses the fact that $\frac{2 |V(H)|}{k} \ge \cB_1 \ge \cB_{2} \ge \cdots$.
	
			Thus at each iteration, $\cB_t$ decreases by a multiplicative factor of $\inparen{ \frac{\lambda}{\delta_0 - \frac{2}{k}} }^2$.
			This is indeed a decrease as long as $\lambda < \delta_0 - \frac{2}{k}$.
			Since $|\cB_0| = |V(H)|$, after $T > \frac{ \log(2 |V(H)|) }{ 2 \log \inparen{ \frac{\delta_0 - \frac{2}{k}}{\lambda} } }$, we have $|\cB_T| < 1$.
			Thus the algorithm terminates after at most $T$ iterations of the outer loop.

			The total number of vertices visited is then 
			\begin{align*}
				\sum_{t=0}^T |\cB_t| 	&\le |V(H)| \cdot \sum_{t=0}^T \inparen{ \frac{ \lambda }{\delta_0 - \frac{2}{k} } }^{2t} \\
										&< |V(H)| \cdot \sum_{t=0}^\infty \inparen{ \frac{ \lambda }{\delta_0 -\frac{2}{k} } }^{2t } \\
										&= |V(H)| \cdot \frac{1}{1 - \frac{\lambda}{\delta_0 - \frac{2}{k} } }
			\end{align*}

			Thus the algorithm runs in time $\bigoh(|V(H)|)$.
		\end{proof}

\section{List recovery capacity theorem}\label{sec:listreccap}
In this appendix, we prove an analog of the list decoding capacity theorem for list recovery.
\begin{theorem}[List recovery capacity theorem]
For every $R > 0$, and $L \geq \ell$, there is some code $\cC$ of rate $R$ over $\F_q$ which is $(R + \eta(\ell,L), \ell, L)$-list recoverable, for any
\[ \eta(\ell, L) \geq \frac{4\ell}{L} \qquad \text{ and } \qquad q \geq \ell^{2/\eta}. \]
Further, for any constants $\eta, R > 0$ and any $\ell$, any code of rate $R$ which is $(R - \eta, \ell, L)$-list recoverable must have $L = q^{\Omega(n)}$.
\end{theorem}

\begin{proof}
The proof follows that of the classical list-decoding capacity theorem.  For the first assertion, consider a random code $\cC$ of rate $R$, and set $\alpha = R + \eta$, for $\eta = \eta(\ell,L)$ as in the statement.  For any set of $L+1$ messages $\Lambda \subset \F^k$, and for any set $T \subset [n]$ of at most $\alpha n$ indices $i$, and for any lists $S_i$ of size $\ell$, the probability that all of the codewords $\cC(x)$ for $x \in \Lambda$ are covered by the $S_i$ is
\[ \PR{ \forall i \in T,x \in \Lambda, \cC(x)_i \in S_i} = \inparen{ \frac{\ell}{q} }^{|T|(L+1)}.\]
Taking the union bound over all choices of $\Lambda, T$, and $S_i$, we see that
\begin{align*}
\PR{ \cC \text{ is not $(\alpha, \ell, L)$-list recoverable } } &\leq 
{q^{Rn} \choose L+1} \inparen{ \sum_{k \geq \alpha n } {q \choose \ell}^k {n \choose k}\inparen{ \frac{\ell}{q}}^{k(L+1)}} \\
&\leq q^{Rn(L+1)} \inparen{ \sum_{k \geq \alpha n} \inparen{ \frac{ eq }{\ell}}^{\ell k} {n \choose k } \inparen{ \frac{\ell}{q}}^{k(L+1)}} \\
&\leq q^{Rn(L+1)} \inparen{ \sum_{k \geq \alpha n}{n \choose k} } \inparen{ \frac{\ell}{q} }^{\alpha n(L + 1 - \ell(1 - 1/\ln(q/\ell)))} \\
&\leq\exp_q \inparen{ n \inparen{ (R-\alpha)(L+1) + \alpha \inparen{ \ell + L\log(\ell)/\log(q) } + H(\alpha)/\log(q) } }\\
&\leq \exp_q \inparen{ n \inparen{ -L\eta / 2 + \alpha \ell } } \\
&\leq \exp_q \inparen{ -nL\eta/4 }\\
&<1.
\end{align*}
In particular, there exists a code $\cC$ of rate $R$ which is $(\alpha, \ell, L)$-list recoverable.

For the other direction, fix any code $\cC$ of rate $R$, and choose a random set of $\alpha n$ indices $T \subset [n]$, and for all $i \in T$, choose $S_i \subset \F_q$ of size $\ell$ uniformly at random.  Now, for any fixed codeword $c \in \cC$, 
\begin{align*}
 \PR{ c_i \subset S_i \forall i \in T } &= \inparen{ \frac{\ell}{q} }^{\alpha n}.
\end{align*}
Thus,
\[ \EE\inabs{ \inset{ c \in \cC \suchthat c_i \in S_i \forall i \in T } } = q^{Rn} \inparen{ \frac{\ell}{q}}^{\alpha n} \geq q^{(R - \alpha)n}. \]
In particular, if $\alpha = R - \eta$, then this is $q^{\eta n}$.
\end{proof}

\end{document}